\documentclass[a4,12pt]{article}%
\usepackage[cmex10]{amsmath}
\usepackage{amsmath,amsthm,amssymb}
\usepackage{epic, eepic}
\usepackage{graphicx}
\usepackage{amsfonts}
\usepackage{amssymb}%
\providecommand{\U}[1]{\protect\rule{.1in}{.1in}}
\newtheorem{theo}{Theorem}
\newtheorem{conj}{Conjecture}

\newtheorem{defi}{Definition}
\newtheorem{rema}{Remark}

\newtheorem{lemm}{Lemma}
\newtheorem{lemm2}{Lemma}[section]
\begin{document}

\title{The Quantum Relative Entropy as a Rate Function\\and Information Criteria}
\author{Kazuya Okamura\thanks{e-mail: kazuqi@kurims.kyoto-u.ac.jp}\\
Research Institute of Mathematical Sciences,\\
Kyoto University, Kyoto, 606-8502, Japan }
\maketitle

\begin{abstract}
We prove that the quantum relative entropy is a rate function in large
deviation principle. Next, we define information criteria for quantum states
and estimate the accuracy of the use of them. Most of the results in this
paper are essentially based on Hiai-Ohya-Tsukada theorem.

\end{abstract}


\section{Introduction}

In recent years there is a high level of interest in quantum information
theory among researchers. In particular, they are concentrating their studies on
control and estimation of quantum states in many systems in order to put it to
practical use.

Quantum estimation theory is one of the areas of quantum information theory,
which was established in 1970's \cite{He76,Ho82}. Statistical decision theory
and hypothesis testing for quantum systems were mainly studied in the
beginning stage of this theory. Recently various methods such as information
geometry, asymptotic theory, and Bayesian analysis are studied theoretically
and applied to many experiments. Especially, it is not too much to say that
quantum estimation theory made the monopoly of optimization of measurement.

Our main interest in the paper is asymptotic theory based on large deviation
principle. This is the reason why we use the relative entropy, which plays the
role of a rate function in large deviation principle and is a typical
quasi-distance between probability distributions. The quantum relative entropy
is defined analogusly and becomes a quasi-distance between quantum states.
In early investigations \cite{P87,PRV89,P90}, the possibility that it is a rate function
in large deviation principle was studied in the context of variational principle.
No one however could answer whether this conjecture is correct in a general setting.

A partial answer is given by a quantum version of Stein's lemma in
\cite{HP91,ON00,OH04}, which is conceptually different from the original Stein's lemma \cite{DZ02}
and from that in this paper. By contrast, a quantum version of Sanov's theorem which is widely
accepted has not appeared \cite{HN07}. There is only one proposal given in
\cite{B06}.

We will show that some of classical statistical methods are applicable to
quantum systems, owing to the universality of statistics and information
theory. As a part of this attempt, we will answer the previous question
whether the quantum relative entropy is a rate function. In the concrete, a
quantum version of Stein's lemma and of Sanov's theorem are proved in Section 4.

Next in Section 5, quantum model selection will be discussed. We try to apply
information criteria for quantum states. Information criteria for probability
distributions are usually described in terms of the relative entropy, but we
derive those for quantum states from the quantum relative entropy. In preceding investigations
 \cite{U03,YE11}, information criteria are applied to
probability distributions calculated from POVM (positive operator-valued
measures). Accordingly, the accuracy of the estimation by the use of
information criteria for quantum states has not been discussed.

Mathematical and physical backgrounds in this paper are algebraic quantum
theory, measuring processes, and Hiai-Ohya-Tsukada theorem. We begin with
algebraic quantum theory before stating main results. Its mathematical basis
is the theory of operator algebras to clarify the structures of C$^{\ast}$-algebras
and of von Neumann algebras, the spaces of linear functionals on
them and so on \cite{BR1,T79,T02}. C$^{\ast}$-algebras are mainly used in the
present paper, which describe observables of the system, and equivalence
classes of their $\ast$-representations, called sectors, play the essential
role of classifiers of macroscopic structures in quantum systems. We describe
measuring processes within the framework of algebraic quantum theory
\cite{HO09,Oj10}. We use the (C$^{\ast}$-)tensor product of the algebra of the
system and of an apparatus as the algebra of observables of the composite
system. Measuring interactions are then defined as automorphisms on this
tensor algebra. After these preparations, we can apply Hiai-Ohya-Tsukada
theorem \cite[Theorem 3.2]{HOT83}. This theorem is crucial for relating the quantum
relative entropy with the relative entropy of probability distributions of
measured data. We would like to mention that the methods developed in the
paper arise from the context of Large Deviation Strategy \cite{OO11}.

Finally, in Section 6, we will discuss the quantum $\alpha$-divergence
\cite{AN00} as a quasi-entropy formulated by D. Petz \cite{P85}. See
\cite{AN00} for its fundamental properties. We will show that the quantum
$\alpha$-divergence is applicable in the (C$^{\ast}$-)algebraic setting.
First, for this purpose, the $\alpha$-analogue of Hiai-Ohya-Tsukada theorem
is proved. Secondly, a quantum version of Chernoff bound is proved, which is
different from that in \cite{Au07,NS09,Ha07,N06}. This is a large deviation
type estimate, too. Lastly, we define a Bayesian $\alpha$-predictive state and
prove that this minimizes a risk function constructed of the quantum $\alpha
$-divergence. This result generalizes those in \cite{Ai75,TK05,T06}.

\section{Algebraic Quantum Theory and Measuring Processes}

\subsection{Algebraic Quantum Theory and Preliminaries}

Algebraic quantum theory begins with a C$^{\ast}$-algebra $\mathfrak{A}$ of
observables of the system. A C$^{\ast}$-algebra $\mathfrak{A}$ is a $\ast
$-algebra over $\mathbb{C}$, i.e., an algebra over $\mathbb{C}$ with the
involution $^{\ast}:\mathfrak{A}\rightarrow\mathfrak{A}$ defined by
\begin{equation*}
(aA+bB)^{\ast}=\bar{a}A^{\ast}+\bar{b}B^{\ast},\hspace{3mm} (AB)^{\ast}=B^{\ast}A^{\ast},\hspace{3mm}
A^{\ast\ast}:=(A^{\ast})^{\ast}=A,
\end{equation*}
for $A,B \in\mathfrak{A}, a,b\in\mathbb{C}$, where $\bar{a}$ is the complex conjugate of $a$,
and also is a Banach space with the norm $\Vert\cdot\Vert: \mathfrak{A}\rightarrow\mathbb{R}_{\geq0}%
=\{a\in\mathbb{R}\;|\; a\geq0\}$ which satisfies, for $A,B\in\mathfrak{A},
a,b\in\mathbb{C}$,
\begin{equation*}
\Vert AB \Vert\leq\Vert A \Vert\;\Vert B \Vert,\hspace{3mm} \Vert A^{\ast}A \Vert=\Vert A
\Vert^{2},\hspace{3mm} \Vert A^{\ast}\Vert=\Vert A \Vert.
\end{equation*}
In particular, self-adjoint elements $A=A^{\ast}$ of $\mathfrak{A}$ are called
observables. C$^{\ast}$-algebras $\mathfrak{A}$ are assumed to be unital,
i.e., $\mathfrak{A}$ have the unit $1$ in the paper. Next, a state $\omega$ on
$\mathfrak{A}$ is defined as a normalized positive linear functional:
\begin{equation*}
\omega(1)=1,\hspace{3mm} \omega(A^{\ast}A)\geq 0, \hspace{3mm}\omega(aA+bB)=a\;\omega(A)+b\;\omega(B),
\end{equation*}
for $A,B\in\mathfrak{A}, a,b\in\mathbb{C}$. We denote by $E_{\mathfrak{A}}$
the set of states on $\mathfrak{A}$. As is seen by the definition, the concept
of a state is nothing but a complex-valued output mapping. The
following theorem then holds:

\begin{theo}[GNS construction theorem \cite{BR1,T79}]
Let $\mathfrak{A}$ be a C$^{\ast}$-algebra and $\omega$ a state on $\mathfrak{A}$.
Then, there exist a Hilbert space $\mathcal{H}_{\omega}$ with the inner product $\langle\cdot,
\cdot\rangle$, a vector $\Omega_{\omega}\in\mathcal{H}_{\omega}$, and a $\ast
$-homomorphism $\pi_{\omega}: \mathfrak{A}\rightarrow\text{\boldmath $B$
}(\mathcal{H}_{\omega})$, called a $\ast$-representation (a representation,
for short) of $\mathfrak{A}$, such that $\omega(A)=\langle\Omega_{\omega},
\pi_{\omega}(A)\Omega_{\omega}\rangle$. The triplet $(\pi_{\omega},
\mathcal{H}_{\omega}, \Omega_{\omega})$ is called a GNS representation of
$\mathfrak{A}$ with respect to $\omega$.

Furthermore, two GNS representations $(\pi_{\omega,1}, \mathcal{H}_{\omega,1},
\Omega_{\omega,1})$, $(\pi_{\omega,2}, \mathcal{H}_{\omega,2}, \Omega
_{\omega,2})$ of $\mathfrak{A}$ with respect to $\omega$ are unitarily
equivalent, i.e., there exists a unitary operator $U: \mathcal{H}_{\omega
,1}\rightarrow\mathcal{H}_{\omega,2}$ such that $\pi_{\omega,2}(A)=U\pi
_{\omega,1}(A)U^{\ast}$ for $A\in\mathfrak{A}$ and $\Omega_{\omega,2}=U\Omega_{\omega,1}$.
\end{theo}

We denote by $\mathfrak{A}_{+}^{\ast}$ the set of positive linear fuctionals
on $\mathfrak{A}$. The GNS construction theorem holds for each $\phi
\in\mathfrak{A}_{+}^{\ast}$. Let $\mathcal{H}$ be a Hilbert space. For each
subset $\mathcal{M}$ of $B(\mathcal{H})$, let $\mathcal{M}^{\prime}$
denote the set of all elements in $B(\mathcal{H})$ commuting with every
element in $\mathcal{M}$. A subalgebra $\mathcal{M}$ of $B(\mathcal{H})$ is called a
$\ast$-subalgebra, if $\mathcal{M}$ is invariant under the involution.
Then, a von Neumann algebra $\mathcal{M}$
on $\mathcal{H}$ is a $\ast$-subalgebra of$B(\mathcal{H})$ such that
$\mathcal{M}=\mathcal{M}^{\prime\prime}:=(\mathcal{M}^{\prime})^{\prime}$, and
a factor is a von Neumann algebra $\mathcal{M}$ with trivial center
$\mathfrak{Z}(\mathcal{M}):= \mathcal{M}\cap\mathcal{M}^{\prime}=\mathbb{C}1$.
The center $\mathfrak{Z}(\mathcal{M})$ of the von Neumann algebra
$\mathcal{M}$ is a unique maximal abelian $\ast$-subalgebra of $\mathcal{M}$
commuting with every element of $\mathcal{M}$. It is well known that, for a
subset $\mathcal{J}$ of $B(\mathcal{H})$ which is invariant under the
involution, $\mathcal{J}^{\prime\prime}$ is the smallest von Neumann algebra
containing $\mathcal{J}$. It is called the von Neumann algebra generated by
$\mathcal{J}$. A typical example can be given by the von Neumann algebra generated by
a GNS representation $(\pi_{\omega}, \mathcal{H}_{\omega}, \Omega_{\omega})$
of $\mathfrak{A}$ with respect to $\omega$.
Since the subset $\pi_{\omega}(\mathfrak{A})=\{\pi_{\omega}(A) |
A\in\mathfrak{A}\}$ of $B(\mathcal{H}_{\omega})$ is invariant under the
involution, $\pi_{\omega}(\mathfrak{A})^{\prime\prime}$ is the smallest
von Neumann algbra on $\mathcal{H}_{\omega}$ generated by $\pi_{\omega
}(\mathfrak{A})$.

Let $\pi$ be a representation of $\mathfrak{A}$. A state $\omega$ on
$\mathfrak{A}$ is said to be $\pi$-normal if there exists a normal
state\footnote{A positive linear functional $\rho$ on a von Neumann algebra
$\mathcal{M}$ on $\mathcal{H}$ is said to be normal if there exists a positive
trace class operator $\sigma\in\text{\boldmath$T$}(\mathcal{H})_{+}$ such that
$\rho(A)=\mathrm{Tr}[\sigma A]$ for each $A\in\mathcal{M}$.} $\rho$ on
$\pi(\mathfrak{A})^{\prime\prime}$ such that $\omega(A)=\rho(\pi(A))$ for
every $A\in\mathfrak{A}$. Two representations $\pi_{1},\pi_{2}$ are
quasi-equivalent, denoted by $\pi_{1}\approx\pi_{2}$, if each $\pi_{1}$-normal
state is $\pi_{2}$-normal and vise versa. As a complement, two representations
$\pi_{1},\pi_{2}$ are disjoint, denoted by $\pi_{1}\hspace{0.3em}%
\raisebox{-.41ex}{$\circ$}\hspace{-0.46em}%
\raisebox{0.7ex}{\rotatebox[origin=c]{-90}{--}}\hspace{0.5em}\pi_{2}$, if no
$\pi_{1}$-normal state is $\pi_{2}$-normal and vise versa. A state $\omega$ on
a C$^{\ast}$-algebra $\mathfrak{A}$ is called a factor state if the center
$\mathfrak{Z}_{\omega}(\mathfrak{A}):=\pi_{\omega}(\mathfrak{A})^{\prime}%
\cap\pi_{\omega}(\mathfrak{A})^{\prime\prime}$ of the corresponding von
Neumann algebra $\pi_{\omega}(\mathfrak{A})^{\prime\prime}$ is trivial. We
denote by $F_{\mathfrak{A}}$ the set of factor states on $\mathfrak{A}$.

\begin{defi}
[\cite{Oj03}]A sector of C$^{\ast}$-algebra $\mathfrak{A}$ is defined by a
quasi-equivalence class of factor states of $\mathfrak{A}$.
\end{defi}

A key concept of this paper is the sector defined above. Before mentioning how to
use it, we introduce the central measure $\mu_{\omega}$ of a state $\omega$ on
$\mathfrak{A}$. For a positive linear functional $\omega$ on $\mathfrak{A}$, a
positive regular Borel measure $\mu$ on $E_{\mathfrak{A}}$ is called a
barycentric measure of $\omega$, which is called the barycenter of $\mu$ and
denoted also by $b(\mu)$, if it satisfies
\[
\omega(A)=\int\rho(A)\;d\mu(\rho),\hspace{4mm}A\in\mathfrak{A}.
\]
It is proved in \cite[Theorem 4.1.25]{BR1} that there is a one-to-one
correspondence between the abelian von Neumann subalgebras $\mathcal{B}$ of
$\pi_{\omega}(\mathfrak{A})^{\prime}$ and the orthogonal measures\footnote{See
in \cite[Section 4.1]{BR1}.} $\mu=:\mu_{\mathcal{B}}$ of $\omega$. In the case
that $\mathcal{B}$ is a subalgebra of the center $\mathfrak{Z}_{\omega
}(\mathfrak{A})$, the corresponding orthogonal measure $\mu_{\mathcal{B}}$ is
called a subcentral measure of $\omega$ and satisfies the following property:
for each $\Delta\in\mathcal{B}(E_{\mathfrak{A}})$, GNS representations
$\pi_{\omega_{\Delta}}$ and $\pi_{\omega_{E_{\mathfrak{A}}\backslash\Delta}}$
of $\displaystyle{\omega_{\Delta}=\int_{\Delta}\rho\;d\mu_{\mathcal{B}}(\rho
)}$ and $\displaystyle{\omega_{E_{\mathfrak{A}}\backslash\Delta}%
=\int_{E_{\mathfrak{A}}\backslash\Delta}\!\!\!\!\!\rho\;d\mu_{\mathcal{B}%
}(\rho)}$, respectively, are unitarily equivalent to the subrepresentations of
$\pi_{\omega}$, and are disjoint each other. In particular, $\mu_{\omega}%
:=\mu_{\mathfrak{Z}_{\omega}(\mathfrak{A})}$ is called the central measure of
$\omega$, and pseudosupported\footnote{A measure $\mu$ on a Borel space $K$ is
said to be pseudosupported by an arbitary set $A\subseteq K$ if $\mu(B)=0$ for
all Baire sets $B$ such that $B\cap A=\emptyset$.} by the factor states
$F_{\mathfrak{A}}$.That is, the central measure $\mu_{\omega}$ of a state
$\omega$ decomposes its barycenter into different sectors. It is also seen
that each state $\omega$ has a unique central measure $\mu_{\omega}$ and is
decomposed into sectors by $\mu_{\omega}$. The physical interpretation of a
sector will be discussed in the next subsection.

\subsection{Composite Systems and Measuring Processes}

A measuring process is a vital physical process to characterize quantum systems
and access actual stuations. It is widely accepted that apparatuses interact
with the system under consideration and output values of observables. This
fact is formulated here by a composite system of the system and an apparatus.

Let $\mathfrak{A}$ be a C$^{\ast}$-algebra of the system under consideration,
$\omega$ a state on $\mathfrak{A}$, and $(\pi_{\omega}, \mathcal{H}_{\omega},
\Omega_{\omega})$ a GNS representation of $\mathfrak{A}$ with respect to
$\omega$. We select an observable $A=A^{\ast}\in\pi_{\omega}(\mathfrak{A}%
)^{\prime\prime}$ to be measured and use $C_{0}(A):=\{ f(A)| f\in
C_{0}(Sp(A))\}$ as the algebra of an apparatus measuring $A$, where the
spectrum $Sp(A):=\{\lambda\in\mathbb{C}|\;A-\lambda1\;\text{is}%
\;\text{not invertible\}}(\subseteq\mathbb{R})$ of $A$. Then, the algebra of the
composite system of the system and the apparatus is defined as $\mathfrak{A}%
\otimes_{\mathrm{m}}C_{0}(A)$\footnote{The injective C$^{\ast}$-tensor product
$\mathfrak{A}_{1}\otimes_{\mathrm{m}} \mathfrak{A}_{2}=\mathfrak{A}_{1}%
\otimes_{min} \mathfrak{A}_{2}$ of two C$^{\ast}$-algebras $\mathfrak{A}_{1}$
and $\mathfrak{A}_{2}$ is the completion of the algebraic tensor product
$\mathfrak{A}_{1}\otimes_{alg} \mathfrak{A}_{2}$ by the norm $\Vert\cdot
\Vert_{min}$ defined by $\Vert C\Vert_{min}=\sup\Vert(\pi_{1}\otimes\pi
_{2})(C)\Vert$, $C\in\mathfrak{A}_{1}\otimes_{alg} \mathfrak{A}_{2}$, where
$\pi_{1}$ and $\pi_{2}$ run over all representations of $\mathfrak{A}_{1}$ and
$\mathfrak{A}_{2}$, respectively. See \cite{T79}, in detail.}. The von Neumann
algebra of this algebra in the representation $\pi_{\omega}\otimes id$ is
calculated as
\[
\{(\pi_{\omega}\otimes id)(\mathfrak{A}\otimes_{\mathrm{m}}C_{0}%
(A))\}^{\prime\prime} =\pi_{\omega}(\mathfrak{A})^{\prime\prime}%
\otimes\mathcal{A},
\]
where $\mathcal{A}=C_{0}(Sp(A))^{\prime\prime}$ is a von Neumann subalgebra of
$\text{\boldmath $B$}(\mathcal{H}_{\omega})$, with center
\begin{align*}
\mathfrak{Z}_{\pi_{\omega}\otimes id}(\mathfrak{A}\otimes_{\mathrm{m}}%
C_{0}(A))  &  :=\{(\pi_{\omega}\otimes id)(\mathfrak{A}\otimes
_{\mathrm{m}}C_{0}(A))\}^{\prime\prime} \cap\{(\pi_{\omega}\otimes
id)(\mathfrak{A}\otimes_{\mathrm{m}}C_{0}(A))\}^{\prime}\\
&  = (\pi_{\omega}(\mathfrak{A})^{\prime\prime}\otimes\mathcal{A})\cap
(\pi_{\omega}(\mathfrak{A})^{\prime\prime}\otimes\mathcal{A})^{\prime}\\
&  =(\pi_{\omega}(\mathfrak{A})^{\prime\prime}\otimes\mathcal{A})\cap
(\pi_{\omega}(\mathfrak{A})^{\prime}\otimes\mathcal{A}^{\prime})\\
&  = (\pi_{\omega}(\mathfrak{A})^{\prime\prime}\cap\pi_{\omega}(\mathfrak{A}%
)^{\prime})\otimes(\mathcal{A}\cap\mathcal{A}^{\prime})\\
&  =\mathfrak{Z}_{\omega}(\mathfrak{A})\otimes\mathcal{A}.
\end{align*}
Therefore, the algebra of observables of the apparatus is contained in the
center of the algebra of the composite system. Assume, from now on, that the 
state $\omega$ of the system is a factor for simplicity.

Next, let us consider the definition of measuring interactions and their physical meaning.
In this paper, measuring interactions are defined by automorphisms on the algebra
$\pi_\omega(\mathfrak{A})^{\prime\prime}\otimes \mathcal{A}$ of the composite system.
When we fix a standard form\footnote{A standard form of a von Neumann algebra $\mathcal{M}$
on a Hilbert space $\mathcal{H}$ is a quadtuple $(\mathcal{M},\mathcal{H},J,\mathcal{P})$
consistuting of $\mathcal{M}$,
$\mathcal{H}$, a unitary involution $J$, a self-dual cone $\mathcal{P}$ in $\mathcal{H}$,
which satisfy the followings: (i) $J\mathcal{M}J=\mathcal{M}^\prime$;
(ii) $JAJ=A^\ast$, $A\in\mathfrak{Z}(\mathcal{M})$; (iii) $J\xi=\xi$, $\xi\in\mathcal{P}$;
(iv) $AJAJ\mathcal{P}\subset\mathcal{P}$, $A\in\mathcal{M}$. See \cite{T02}.}
 of the von Neumann algebra $\pi_\omega(\mathfrak{A})^{\prime\prime}\otimes \mathcal{A}$,
any automorphism $\alpha$ is unitarily implemented,
i.e., there exists a unitary operator $U$ such that $\alpha(A)=U^{\ast}A U$,
$A\in \pi_\omega(\mathfrak{A})^{\prime\prime}\otimes \mathcal{A}$.
Owing to this result, we can reproduce unitary opearators on a Hilbert space 
as measuring interactions in general theory of measuring processes in Hilbert space formalism.
Let $\alpha$ be an automorphism on a von Neumann algebra $\mathcal{M}$.
We define $\alpha^\ast:E_\mathcal{M}\rightarrow E_\mathcal{M}$ by $(\alpha^\ast\omega)(B)=\omega(\alpha(B))$
for $B\in\mathcal{M}$.

The most important and fundamental measuring interaction is that of ideal measurement of $A$ denoted by $\alpha(W)$, which is to be combined with the concept of a neutral position, vital for the description of function of apparatuses \cite{HO09}. A neutral position, denoted by $m_\mathcal{A}$, is the 
 \textquotedblleft universal\textquotedblright\ initial state of the apparatus
and correponds to the macroscopically stable position of the measuring pointer realized when the apparatus is isolated (for the detail mathematical formulation, see \cite{HO09}). $\alpha(W)$ is an automorphism on
$\pi_\omega(\mathfrak{A})^{\prime\prime}\otimes \mathcal{A}$ such that
\begin{eqnarray}\label{Born}
 \omega_{(A,\alpha(W))}(B)&:=&\alpha(W)^\ast(\tilde{\omega} \otimes m_\mathcal{A})((\pi_\omega\otimes id)(B))
 \nonumber\\
 &=&\int (\omega_{a}\otimes_{\mathrm{m}} \delta_a)(B)\;P(A\in da\Vert \omega),
\end{eqnarray}
for $B\in \mathfrak{A}\otimes_{\mathrm{m}}C_0(A)$,
where $P(A\in \Delta\Vert \omega)=\langle \Omega_\omega, E^A(\Delta)\Omega_\omega\rangle$,
$E^A$ the spectral measure of $A$ such that $\displaystyle{A=\int a\;dE^A(a)}$, $\Delta\in\mathcal{B}(\mathbb{R})$.
$\omega_a=\tilde{\omega}_a\circ \pi_\omega$ such that $\tilde{\omega}_a(f(A))=f(a)$,
$\delta_a$ the eigenstate on $C_0(A)$ for $a\in Sp(A)$.
The interaction $\alpha(W)$ of ideal measurement should be regarded as
a trigger of a perfect correlation \cite{Oz06} between the state of the system and of the apparatus.

Then, Eq.(\ref{Born}) can be rewritten as
\begin{eqnarray}\label{CM}
 \omega_{(A,\alpha(W))}(B)\!\!\!&=&\!\!\!\int \rho(B)\;d\mu_{\omega_{(A,\alpha)}}(\rho),  \\
\text{supp}\;\mu_{\omega_{(A,\alpha(W))}}
\!\!\!&=&\!\!\!\overline{\{\omega_{a}\otimes_{\mathrm{m}} \delta_a|a\in Sp(A)\}},\nonumber
\end{eqnarray}
for $B\in \mathfrak{A}\otimes_{\mathrm{m}}C_0(A)$.
This is because the algebra $\mathfrak{A}\otimes_{\mathrm{m}}C_0(A)$ of the composite system
has center $\mathbb{C}1\otimes \mathcal{A}$, and a family of the eigenvalues of an observable $A$ become an index for classifying sectors.
This result shows us that the central measure of a state of
the composite system after measurement play the role of a probability measure appearing
in Born rule \cite{OOS11}. We can generalize the above argument for
any automorphism $\alpha$ on $\pi_\omega(\mathfrak{A})^{\prime\prime}\otimes \mathcal{A}$,
and discuss the optimization of measurement.
\begin{eqnarray}
 \omega_{(A,\alpha)}(B)&:=&\alpha^\ast(\tilde{\omega} \otimes m_\mathcal{A})((\pi_\omega\otimes id)(B))
 \nonumber\\
 &=&\!\!\!\int \rho(B)\;d\mu_{\omega_{(A,\alpha)}}(\rho),  \\
\text{supp}\;\mu_{\omega_{(A,\alpha(W))}}
\!\!\!&=&\!\!\!\overline{\{\omega_{\alpha,a}\otimes_{\mathrm{m}} \delta_a|a\in Sp(A)\}}.\nonumber
\end{eqnarray}
for $B\in \mathfrak{A}\otimes_{\mathrm{m}}C_0(A)$, where, for each $a\in Sp(A)$,
$\omega_{\alpha,a}$ depending on $\alpha$ is the state of the system when the value $a\in Sp(A)$
is measured in an apparatus.
Therefore, we use the central measure of a state of the composite system after measuring interaction.
If $\omega$ above is not a factor, we should replace the central measure $\mu_{\omega_{(A,\alpha)}}$
by the subcentral measure of $\omega_{(A,\alpha)}$ corresponding to the abelian von Neumann
subalgebra $\mathbb{C}1\otimes \mathcal{A}$.

We close this section with discussion on the concept of a sector.
A sector can be interpreted as a physical origin of macroscopic indicators of classification,
and disjointness among different sectors represents a 
\textquotedblleft conditional\textquotedblright\ stability of macroscopic structures \cite{Oj05,Oj03}.
As a typical example, sectors in the composite system of the system under consideration with an apparatus
are classified by the eigenvalues of an observable to be measured, which was already discussed.
We would like to emphasize that quantum measurement theory requires searching
nontrivial intersectorial structure in each sector of the system.
On the other hand, non-factor states of the system correspond to such physical situations that order parameters
like temparature are assumed before any measurement. This case is often overlooked.

\section{Hiai-Ohya-Tsukada Theorem and its Application}
\subsection{Hiai-Ohya-Tsukada Theorem}
We denote by $M_{1}(\Omega)$ the space of probability measures on a measurable space
$(\Omega,\mathcal{F})$. We define the relative entropy of the probability measure $\nu\in M_1(\Omega)$
with respect to $\mu\in M_1(\Omega)$ as
\begin{equation}
D(\nu\Vert\mu)=\left\{
\begin{array}[c]{c}
\displaystyle{\int d\nu(\rho)\log\dfrac{d\nu}{d\mu}(\rho)\;\;\;\;(\nu\ll\mu)}\\
+\infty\;\;\;\;(\mathrm{otherwise}).
\end{array}
\right.
\end{equation}
If there exists a measure $\sigma$ on $\Omega$ such that $\nu,\mu \ll \sigma$,
$D(\nu\Vert\mu)$ is also denoted by $D(q\Vert p)$ where $q:=\dfrac{d\nu}{d\sigma}$ and $p:=\dfrac{d\mu}{d\sigma}$.

Let us explain the formulations of relative entropy due to Araki and to Uhlmann \cite{Ar77,U77,HOT83}.
Let $(\mathcal{M},\mathcal{H},J,\mathcal{P})$ be a standard form of a von Neumann algebra $\mathcal{M}$,
and $\varphi, \psi$ be normal states on $\mathcal{M}$. There exist unique $\Phi,\Psi\in\mathcal{P}$ such that
$\varphi(A)=\langle \Phi,A\Phi \rangle, \psi(A)=\langle \Psi,A\Psi \rangle$ for all $A\in\mathcal{A}$.
The operator $S_{\Phi,\Psi}$ with the domain
$Dom(S_{\Phi,\Psi})=\mathcal{M}\Psi+(1-s^\mathcal{M^\prime}(\Psi))\mathcal{H}$ is defined by
\begin{equation*}
S_{\Phi,\Psi}(A\Psi+\Omega)=s^\mathcal{M}(\Psi)A^\ast\Phi,\hspace{3mm}A\in\mathcal{M},s^\mathcal{M^\prime}(\Psi)\Omega=0,
\end{equation*}
where $s^\mathcal{M}(\Psi)$ is $\mathcal{M}$-support of $\Psi$, i.e., the minimal projection $E$ in $\mathcal{M}$
such that $(1-E)\Psi=0$. $S_{\Phi,\Psi}$ is seen to be a closable operator. Then the relative modular operator
$\Delta_{\Phi,\Psi}$ is defined by $\Delta_{\Phi,\Psi}=(S_{\Phi,\Psi})^\ast \overline{S_{\Phi,\Psi}}$, and let
$\Delta_{\Phi,\Psi}=\displaystyle{\int_{-\infty}^\infty \lambda dE_{\Phi,\Psi}(\lambda)}$ be the spectral
decomposition of $\Delta_{\Phi,\Psi}$. Araki's relative entropy $S(\varphi\Vert\psi)_{\mathrm{Araki}}$ is defined by
\begin{equation*}
S(\varphi\Vert\psi)_{\mathrm{Araki}}=
\left\{
\begin{array}{ll}
\displaystyle{\int_{-\infty}^\infty \log\lambda\;d\langle \Phi,E_{\Phi,\Psi}(\lambda)\Phi \rangle},&\\
\quad\quad\quad(s(\varphi)\leq s(\psi)),& \\
+\infty,\quad(otherwise),&
\end{array}
\right.
\end{equation*}
where $s(\varphi)$ is the minimal projection $E$ in $\mathcal{M}$ such that $\varphi(1-E)=0$ and
called the support of $\varphi$.

We then define Uhlmann's relative entropy.
For two seminorms $p$ and $q$ on a complex linear space $\mathcal{L}$,
the quadratical mean $QM(p,q)$ is defined by
\begin{equation*}
QM(p,q)(x)=\sup_{\alpha\in \mathcal{S}(p,q)} \alpha(x,x)^\frac{1}{2},\hspace{3mm}x\in\mathcal{L},
\end{equation*}
where $\mathcal{S}(p,q)$ is the set of all positive hermitian forms $\alpha$ on $\mathcal{L}$
satisfying $|\alpha(x,y)|\leq p(x)q(y)$ for all $x,y\in\mathcal{L}$.
A fuction $[0,1]\ni t\mapsto p_t$ whose values are seminorms on $\mathcal{L}$ is called a quadratical interpolation
from $p$ to $q$ if it satisfies the following conditions:
(i) for each $x\in\mathcal{L}$ the function $t\mapsto p_t(x)$ is continuous;
(ii) it satisfies the following properties:
\begin{eqnarray*}
p_t &=&QM(p_{t_1},p_{t_2}),\hspace{2mm}t=\frac{t_1+t_2}{2},t_1,t_2\in[0,1], \\
p_{\frac{1}{2}} &=&QM(p,q), \\
p_{\frac{t}{2}} &=&QM(p,p_t),\hspace{2mm}t\in[0,1], \\
p_{\frac{1+t}{2}} &=&QM(p_t,q),\hspace{2mm}t\in[0,1].
\end{eqnarray*}
Furthermore, for each positive hermitian forms $\alpha$ and $\beta$ there exists a unique function
$[0,1]\ni t\mapsto QF_t(\alpha,\beta)$ whose values are positive hermitian forms on $\mathcal{L}$
such that the function $p_t$ defined by $p_t(x)=QF_t(\alpha,\beta)(x,x)^\frac{1}{2}$ for each $x\in\mathcal{L}$
is the quadratical interpolation from $\alpha(x,x)^\frac{1}{2}$ to $\beta(x,x)^\frac{1}{2}$.
The relative entropy functional $S(\alpha\Vert\beta)(x)$ of $\alpha$ and $\beta$ is defined by
\begin{equation*}
S(\alpha\Vert\beta)(x)=-\liminf_{t\rightarrow +0}
\frac{QF_t(\alpha,\beta)(x,x)-\alpha(x,x)}{t}.
\end{equation*}

Let $\mathcal{A}$ be a $\ast$-algebra, and $\varphi,\psi$ be a positive linear functionals on $\mathcal{A}$.
The Uhlmann's relative entropy $S(\varphi\Vert\psi)_{\mathrm{Uhlmann}}$ is defined by
\begin{equation*}
S(\varphi\Vert\psi)_{\mathrm{Uhlmann}}=S(\varphi^R\Vert\psi^L)(1),
\end{equation*}
where $\varphi^R$ and $\psi^L$ are the positive hermitian forms given by $\varphi^R(A,B)=\varphi(BA^\ast)$,
$\psi^L(A,B)=\psi(A^\ast B)$.

In \cite{HOT83}, the following three important theorems are proved:
\begin{theo}
For any positive linear functionals $\varphi,\psi$ on a von Neumann algebra $\mathcal{M}$,
$S(\varphi\Vert\psi)_{\mathrm{Uhlmann}}=S(\varphi\Vert\psi)_{\mathrm{Araki}}$.
\end{theo}
Therefore, we do not need to distinguish Araki's relative entropy and Uhlmann's one,
and call them the quantum relative entropy. They are denoted by $S(\varphi\Vert\psi)$.
\begin{theo}
Let $\varphi,\psi$ be positive linear functionals on a C$^\ast$-algebra $\mathfrak{A}$
and $\pi$ be a nondegenerate representation of $\mathfrak{A}$ on a Hilbert space.
If there are the normal extentions $\tilde{\varphi},\tilde{\psi}$ of $\varphi,\psi$ to
$\pi(\mathfrak{A})^{\prime\prime}$ such that
$\varphi(A)=\tilde{\varphi}(\pi(A)),\psi(A)=\tilde{\psi}(\pi(A))$, then
$S(\tilde{\varphi}\Vert\tilde{\psi})=S(\varphi\Vert\psi)$.
\end{theo}
We take as $\pi$ a GNS representation $\pi_{\varphi+\psi}$ corresponding to $\varphi+\psi$ in this paper.
The following theorem is our main concern in this section.

\begin{theo}[Hiai-Ohya-Tsukada Theorem\cite{HOT83}]\label{HOT}
Let $\varphi,\psi$ be states on $\mathfrak{A}$ with barycentric measures $\mu,\nu$, respectively.
If there exists a subcentral measure $m$ such that $\mu,\nu\ll m$, then $S(\varphi\Vert\psi)=D(\mu\Vert\nu)$.
\end{theo}

By this theorem we can calculate the quantum relative entropy as the measure-theoretical relative entropy.

\subsection{An Application of Theorem 4}
Let $\mathfrak{A}$ be a C$^\ast$-algebra,
$\omega_1,\omega_2$ states on $\mathfrak{A}$, and $(\pi_\omega, \mathcal{H}_\omega, \Omega_\omega)$
a GNS representation of $\mathfrak{A}$ with respect to $\omega=\omega_1+\omega_2$.
We define the set of measurements comparing $\omega_1$ with $\omega_2$ as follows:
\begin{eqnarray*}
\mathrm{M}(\omega_1\prec\omega_2) \!\!\!&=&\!\!\!\bigcup_{A=A^\ast\in\pi_{\omega}(\mathfrak{A})^{\prime\prime}}
\mathrm{M}(A; \omega_1\prec\omega_2), \\
\mathrm{M}(A; \omega_1\prec\omega_2) \!\!\!&=&\!\!\!
\{(A,\alpha) \in \pi_{\omega}(\mathfrak{A})^{\prime\prime}\times Aut(\pi_\omega(\mathfrak{A})^{\prime\prime}\otimes \mathcal{A})|\\
&&\mathcal{A}=C_0(A)^{\prime\prime}, \mu_{\omega_{1,(A,\alpha)}}\ll \mu_{\omega_{2,(A,\alpha)}} \},
\end{eqnarray*}
where
\begin{eqnarray*}
\omega_{1,(A,\alpha)} &=& \alpha^\ast(\tilde{\omega}_1\otimes m_\mathcal{A})\circ(\pi_\omega\otimes id),\\
\omega_{2,(A,\alpha)} &=&\alpha^\ast(\tilde{\omega}_2\otimes m_\mathcal{A})\circ(\pi_\omega\otimes id),
\end{eqnarray*}
and the C$^\ast$-algebra of the composite system is $\mathfrak{A}\otimes_{\mathrm{m}}C_0(A)$.
In order to give the validity of the definition of $\mathrm{M}(\omega_1\prec\omega_2)$, we prepare the next lemma.
\begin{lemm}
Let $\omega_1,\omega_2$ be states on a C$^\ast$-algebra $\mathfrak{A}$, $\alpha\in Aut(\mathfrak{A})$.
Then, $S(\omega_1\Vert\omega_2)=S(\alpha^\ast\omega_1\Vert\alpha^\ast\omega_2)$.
\end{lemm}
\begin{proof}
For a CP map $\Lambda$ on $\mathfrak{A}$, it holds that
\begin{equation*}
S(\Lambda^\ast\omega_1\Vert\Lambda^\ast\omega_2)\leq S(\omega_1\Vert\omega_2),
\end{equation*}
which is a well-known fact \cite{HOT83,OP93}. Thus,
\begin{equation*}
S(\alpha^\ast\omega_1\Vert\alpha^\ast\omega_2)\leq S(\omega_1\Vert\omega_2).
\end{equation*}
Since $\alpha^{-1}$ is also an automorphism on $\mathfrak{A}$,
\begin{equation*}
S(\omega_1\Vert\omega_2)=S((\alpha^{-1})^\ast\alpha^\ast\omega_1\Vert(\alpha^{-1})^\ast\alpha^\ast\omega_2)
\leq S(\alpha^\ast\omega_1\Vert\alpha^\ast\omega_2).
\end{equation*}
The lemma is proved.
\end{proof}

If $\mathrm{M}(\omega_1\prec\omega_2)\neq\emptyset$, then for a $(A,\alpha)\in \mathrm{M}(\omega_1\prec\omega_2)$,
\begin{eqnarray*}
\infty&>&D(\mu_{\omega_{1,(A,\alpha)}}\Vert
 \mu_{\omega_{2,(A,\alpha)}}) \\
 &=&S(\alpha^\ast(\tilde{\omega}_1\otimes m_\mathcal{A})\circ(\pi_\omega\otimes id)\Vert
 \alpha^\ast(\tilde{\omega}_2\otimes m_\mathcal{A})\circ(\pi_\omega\otimes id)) \\
 &=&S(\alpha^\ast(\tilde{\omega}_1\otimes m_\mathcal{A})\Vert
 \alpha^\ast(\tilde{\omega}_2\otimes m_\mathcal{A})) \\
 &=&S(\tilde{\omega}_1\otimes m_\mathcal{A}\Vert\tilde{\omega}_2\otimes m_\mathcal{A})
 =S(\tilde{\omega}_1\Vert\tilde{\omega}_2)=S(\omega_1\Vert\omega_2).
\end{eqnarray*}
The proof of this equality is essentially based on the use of $(A,\alpha)\in\mathrm{M}(\omega_1\prec\omega_2)$,
Theorem 2 and Lemma 1. The central measures $\mu_{\omega_{1,(A,\alpha)}}, \mu_{\omega_{2,(A,\alpha)}}$
of states $\omega_1, \omega_2$, respectively, is specified by measured data, since
Born rule is equivalent to the central measure of the state of the composite system after measurement.
We conclude that we can examine the quantum relative entropy $S(\omega_1\Vert\omega_2)$ of $\omega_1$
with respect to $\omega_2$ statistically, whenever $\mathrm{M}(\omega_1\prec\omega_2)\neq\emptyset$.

By the way, we may conjecture the following statement:
\begin{conj}
If $S(\omega_1\Vert\omega_2)<\infty$, then $\mathrm{M}(\omega_1\prec\omega_2)\neq\emptyset$.
\end{conj}
This is, of course, not trivial mathematically.
The strong condition $S(\omega_1\Vert\omega_2)<\infty$ should be regarded as 
a physically affirmative statement that $\omega_1$ and $\omega_2$ are so similar
that the former can be compared with the latter
by the quantum relative entropy $S(\omega_1\Vert\omega_2)$.
Thus it is expected that the quantum relative entropy corresponds to 
the measure-theoretical relative entropy using specific probability measures.
In the rest of the paper, we assume a condition such as $\mathrm{M}(\omega_1\prec\omega_2)\neq\emptyset$.

\section{Large Deviation Type Estimate}
In this section, two estimates based on large deviation principle are discussed.
One of them is Stein's lemma and the other Sanov's theorem.
These are regarded as the standard procedure to give the statistical meaning to the relative entropy.
We show here that in quantum systems the quantum relative entropy has the same meaning as
the measure-theoretical relative entropy has.

\subsection{Another Version of Quantum Stein's Lemma}
Assume that $\mathrm{M}(\omega_1\prec\omega_2)\neq\emptyset$, and
for a $(A,\alpha)\in \mathrm{M}(\omega_1\prec\omega_2)$
the Radon-Nikodym derivative $\dfrac{d\mu_{\omega_{1,(A,\alpha)}}}{d\mu_{\omega_{2,(A,\alpha)}}}$
is strictly positive.
We are now ready for hypothesis testing.
\begin{defi}
A (decision) test $\mathcal{T}$ is a sequence of measurable functions
$T_n:(\mathrm{supp}\;\mu_{\omega_{2,(A,\alpha)}})^n$ $\rightarrow$ $\{0,1\}$ with interpretation
$\{T_n=0\}=\{H_0\;\text{is}\;\text{accepted}\}$ and $\{T_n=1\}=\{H_1\;\text{is}\;\text{accepted}\}$.
If $\mathcal{T}=(T)$, we do not distinguish $\mathcal{T}$ and $T$.
\end{defi}
We define the error probabilities
\begin{eqnarray*}
\alpha_n^{(A,\alpha)}(T_n) &=&P_{\omega_1}^{(A,\alpha)}(T_n=1)=P_{\omega_1}^{(A,\alpha)}(H_0\;\text{is}\;\text{rejected}), \\
\beta_n^{(A,\alpha)}(T_n) &=&P_{\omega_2}^{(A,\alpha)}(T_n=0)=P_{\omega_2}^{(A,\alpha)}(H_1\;\text{is}\;\text{rejected}).
\end{eqnarray*}
of the first kind and of the second kind,
where $P_{\omega_i}^{(A,\alpha)}$ is the countable product probability measure of $\mu_{\omega_{i,(A,\alpha)}}$
for $i=1,2$.
Then, we define the log-likelihood ratio
\begin{equation*}
X_j=-\log \frac{d\mu_{\omega_{1,(A,\alpha)}}}{d\mu_{\omega_{2,(A,\alpha)}}}(\rho_j)
\end{equation*}
for $\tilde{\rho}=$, and the normalized log-likelihood ratio
\begin{equation*}
S_n=\frac{1}{n}\sum_{j=1}^n X_j.
\end{equation*}
A test $\mathcal{T}=(\tilde{T}_n:(\mathrm{supp}\;\mu_{\omega_{2,(A,\alpha)}})^n\rightarrow \{0,1\})$ defined by
\begin{equation*}
\tilde{T}_n(\rho^n)=
\left\{
\begin{array}{ll}
0, &\quad\left(S_n\leq \eta_n\right),  \\
1, &\quad(otherwise),
\end{array}
\right.
\end{equation*}
is called the Neymann-Pearson test, and denoted also by $\mathcal{T}^\eta=(T^{\eta_n}_n)$,
where $\eta=(\eta_n)$ is a real sequence.
\begin{lemm}\label{NP}
For any $0\leq \gamma\leq 1$, there exists a $\eta=(\eta_n)$ such that
$\alpha_n^{(A,\alpha)}(\tilde{T}_n^{\eta_n})\leq\gamma$,
and any other test $\mathcal{T}=(T_n)$ satisfying $\alpha_n^{(A,\alpha)}(T_n)\leq\gamma$
must satisfy $\beta_n^{(A,\alpha)}(T_n)\geq\beta_n^{(A,\alpha)}(\tilde{T}^{\eta_n}_n)$.
\end{lemm}
See \cite{CT91,LR05} for proof.
It is well known that both of the error probabilities
cannot tend to zero simultaneously. Thus we consider the next quantity.
\begin{equation*}
\beta_n^{(A,\alpha)}(\varepsilon)=\inf
\{\beta_n^{(A,\alpha)}(T_n)\;|\;T_n\;:\;\text{test},\;\alpha_n^{(A,\alpha)}(T_n)<\varepsilon\}
\end{equation*}
This is the infimum of $\beta_n^{(A,\alpha)}(T_n)$ among all tests $T_n$
with $\alpha_n^{(A,\alpha)}(T_n)<\varepsilon$.
It is seen by Lemma \ref{NP} that a decision test attaining the infimum is the Neymann-Pearson test
$\mathcal{T}^\eta=(T^{\eta_n}_n)$, where $\eta=(\eta_n)$ is determined by $\epsilon$.
The following theorem is then proved, which is another version of quantum Stein's lemma.
\begin{theo}
For any $\varepsilon<1$,
\begin{equation}
\lim_{n\rightarrow \infty}\frac{1}{n}\log\beta_{n}^{(A,\alpha)}(\varepsilon)=-S(\omega_1\Vert\omega_2).
\end{equation}
\end{theo}
\begin{proof}
By \cite[Lemma 3.4.7 and Exercise 3.4.17]{DZ02},
\begin{equation*}
\lim_{n\rightarrow \infty}\frac{1}{n}\log\beta_{n}^{(A,\alpha)}(\varepsilon)=
-D(\mu_{\omega_{1,(A,\alpha)}}\Vert\mu_{\omega_{2,(A,\alpha)}}).
\end{equation*}
Then,
$D(\mu_{\omega_{1,(A,\alpha)}}\Vert\mu_{\omega_{2,(A,\alpha)}})=S(\omega_1\Vert\omega_2)$.

\end{proof}
This theorem holds for all C$^\ast$-algebras, that is,
for all quantum systems, if two states $\varphi,\psi$ are comparable each other, i.e.,
$\mathrm{M}(\varphi\prec\psi)\neq\emptyset$ or $\mathrm{M}(\psi\prec\varphi)\neq\emptyset$.
Therefore, we expect that the asymptotic theory of classical hypothesis testing is applicable to the quantum case.

\subsection{A Quantum Version of Sanov's Theorem}
Let $\mathfrak{A}$ be a C$^\ast$-algebra, $\omega$ be a state on $\mathfrak{A}$,
$(\pi_\omega,\mathcal{H}_\omega,\Omega_\omega)$ be a GNS representation of $\mathfrak{A}$
with respect to $\omega$. We use $\mathfrak{A}\otimes_{\mathrm{m}}C_0(A)$, where
$A=A^\ast\in\pi_\omega(\mathfrak{A})^{\prime\prime}$, as the algebra of the composite system
of the system and an apparatus measuring $A$, and
$\alpha\in Aut(\pi_\omega(\mathfrak{A})^{\prime\prime}\otimes \mathcal{A})$,
$\mathcal{A}=C_0(A)^{\prime\prime}$, as a measuring interaction.

We denote by $\mathcal{B}^{w}(M_{1}(E_{\mathfrak{A}\otimes_{\mathrm{m}} C_0(A)}))$
the Borel $\sigma$-field on $M_{1}(E_{\mathfrak{A}\otimes_{\mathrm{m}} C_0(A)})$ generated by
the weak topology \footnote{See \cite{Cs06} and \cite[p.261]{DZ02}.}. 
For any $\tilde{\rho}=(\rho_1,\rho_2,\cdots)\in(\mathrm{supp}\;\mu_{\omega_{(A,\alpha)}})^{\mathbb{N}}$ and
$\Delta\in\mathcal{B}(\mathrm{supp}\;\mu_{\omega_{(A,\alpha)}})$, we define random variables
$Y_j(\tilde{\rho})=\rho_j$ for $j=1,2,\cdots$, the empirical measures
\begin{equation}
L_{n}^{(A,\alpha)}(\tilde{\rho},\Delta)=\frac{1}{n}\sum_{j=1}^{n}\delta_{Y_{j}(\tilde{\rho})}(\Delta),
\end{equation}
and
\begin{equation}
Q^{(A,\alpha)}_{n}(\Gamma)=P_{\omega}^{(A,\alpha)}(L_{n}^{(A,\alpha)}\in\Gamma),
\end{equation}
for $\Gamma\in\mathcal{B}^{w}(M_{1}(E_{\mathfrak{A}\otimes_{\mathrm{m}} C_0(A)}))$
such that $\{L_{n}^{(A,\alpha)}\in\Gamma\}$ is measurable, where $P_{\omega}^{(A,\alpha)}$
denotes the countable product measure of $\mu_{\omega_{(A,\alpha)}}$.
Then we state the main theorem in this subsection,
the earlier version of which has appeared in \cite{OO11}.
\begin{theo}
$Q^{(A,\alpha)}_n$ satisfies the large deviation principle
with the good rate function $S(b(\cdot)\Vert \omega_{(A,\alpha)})$:
\begin{equation}
-S(b(\Gamma^{o})\Vert\omega_{(A,\alpha)})\leq \liminf_{n\rightarrow\infty}\frac{1}{n}\log Q_{n}^{(A,\alpha)}(\Gamma)
\leq \limsup_{n\rightarrow\infty}\frac{1}{n}\log Q_{n}^{(A,\alpha)}(\Gamma)\leq
-S(b(\overline{\Gamma})\Vert\omega_{(A,\alpha)})
\end{equation}
for $\Gamma\in\mathcal{B}^{w}(M_{1}(E_{\mathfrak{A}\otimes_{\mathrm{m}} C_0(A)}))$
such that $\{L_{n}^{(A,\alpha)}\in\Gamma\}$ is measurable, where
\begin{equation*}
S(b(\Gamma)\Vert\omega_{(A,\alpha)}):=
\inf\{S(b(\nu)\Vert\omega_{(A,\alpha)})|\nu\in\Gamma,\nu \ll \mu_{\omega_{(A,\alpha)}}\}.
\end{equation*}
In the case that, for $\Gamma\in\mathcal{B}^{w}(M_{1}(E_{\mathfrak{A}\otimes_{\mathrm{m}} C_0(A)}))$,
$\{\nu\in\Gamma^{o}|\nu \ll \mu_{\omega_{(A,\alpha)}}\}$
and $\{\nu\in\overline{\Gamma}|\nu \ll \mu_{\omega_{(A,\alpha)}}\}$
are empty, $S(b(\Gamma^{o})\Vert\omega_{(A,\alpha)})$ and
$S(b(\overline{\Gamma})\Vert\omega_{(A,\alpha)})$ are
defined as infinity, respectively.
\end{theo}
\begin{proof}
This theorem is easily proved by \cite[Sanov's Theorem]{Cs06} and Theorem \ref{HOT}.
\end{proof}
If $\mathfrak{A}\otimes_{\mathrm{m}} C_0(A)$ is separable,
then $E_{\mathfrak{A}\otimes_{\mathrm{m}} C_0(A)}$ becomes a compact metric space whose metric is defined by
\begin{equation*}
d(\rho_1,\rho_2)=\sum_{n=1}^\infty \frac{1}{2^n}\frac{|\rho_1(A_j)-\rho_2(A_j)|}{\Vert A_j\Vert},
\end{equation*}
where the set $\{A_j\neq 0|j\in\mathbb{N}\}$ is a dense subset of $\mathfrak{A}\otimes_{\mathrm{m}} C_0(A)$.
In this case, for all $\Gamma\in\mathcal{B}^{w}(M_{1}(E_{\mathfrak{A}\otimes_{\mathrm{m}} C_0(A)}))$,
$\{L_{n}^{(A,\alpha)}\in\Gamma\}$ is measurable. It is known that
${B}^{w}(M_{1}(E_{\mathfrak{A}\otimes_{\mathrm{m}} C_0(A)}))$ is equal to
${B}^{cy}(M_{1}(E_{\mathfrak{A}\otimes_{\mathrm{m}} C_0(A)}))$,
if $\mathfrak{A}\otimes_{\mathrm{m}} C_0(A)$ is separable,
where ${B}^{cy}(M_{1}(E_{\mathfrak{A}\otimes_{\mathrm{m}} C_0(A)}))$ is defined as follows.
We denote by $B(E_{\mathfrak{A}\otimes_{\mathrm{m}} C_0(A)})$ the vector space of all bounded
Borel measurable functions on $E_{\mathfrak{A}\otimes_{\mathrm{m}} C_0(A)}$.
For $\phi\in B(E_{\mathfrak{A}\otimes_{\mathrm{m}} C_0(A)})$,
let $\tau_{\phi}:M_{1}(E_{\mathfrak{A}\otimes_{\mathrm{m}} C_0(A)})\rightarrow\mathbb{R}$
be defined by $\tau_{\phi}(\nu)=\langle\phi,\nu\rangle=\int \phi\;d\nu$.
We denote by
$\mathcal{B}^{cy}(M_{1}(E_{\mathfrak{A}\otimes_{\mathrm{m}} C_0(A)}))$ the $\sigma$-field of cylinder
sets on $M_{1}(E_{\mathfrak{A}\otimes_{\mathrm{m}} C_0(A)})$, i.e., the smallest $\sigma$-field that makes all
$\{\tau_{\phi}\}$ measurable \footnote{See \cite[p.263]{DZ02}}.

When one of the measures $\mu$ attaining the infimum of upper or lower bound is an element of
$\Omega^{(A,\alpha)}=\{\nu\in E_{\mathfrak{A}\otimes_{\mathrm{m}} C_0(A)}|b(\nu)=\psi_{(A,\alpha)},
\psi\in E_{\mathfrak{A}},\psi : \pi_\omega\text{-normal}\}$,
there exists at least one state $\psi$
such that $S(b(\nu)\Vert \omega_{(A,\alpha)})=S(\psi\Vert\omega)$.
We can therefore reach the state of the system by the use of measured data.
\begin{rema}
Optimization of measurement is not considered in this subsection.
The setting of Sanov's theorem requires only the knowledge of the methods for collecting and accumulating data.
As a result, we should use a measurement $(A,\alpha)$ which is analyzed in detail in a quantum system
under consideration.
\end{rema}

\section{Quantum Model Selection}
First, we define a (parametric) model of states.
\begin{defi}
A family of states $\{\omega_{\theta}|\theta\in\Theta\}$ on $\mathfrak{A}$ is called a (statistical) model
if it satisfies the following conditions:\\
(i) $\Theta$ is a compact subset of $\mathbb{R}^d$ $(d\in\mathbb{N})$:\\
(ii) there exist $\theta_1,\cdots,\theta_g\in\Theta$ such that, for every $\theta\in\Theta$, $\omega_{\theta}$ is
$\pi$-normal, where $\pi=\pi_{\omega_{\theta_1}+\cdots+\omega_{\theta_g}}$.
\end{defi}
Assume that each model in this section is factor for simplicity.
A predictive state is a quantum version of predictive (probability) distribution and is a function,
which is constructed of a model of states, from data into states. In this section,
our purpose is to define a quantum version of information criteria in order to choose the best predictive state
from many models. We then define the following measurement.
\begin{defi}
For $N\in\mathbb{N}$, $N$ different models
$\{\omega_{1,\theta_1}|\theta_1\in\Theta_1\},\cdots, \{\omega_{N,\theta_N}|\theta_N\in\Theta_N\}$
are comparable if they satisfy the following conditions:\\
(i) there exists $\varphi\in E_{\mathfrak{A}}$ such that
$\mathrm{M}(\varphi\prec\{\omega_{j,\theta_j}\}_{j=1,\cdots,N})\neq\emptyset$;\\
(ii) For a $(A,\alpha)\in \mathrm{M}(\varphi\prec\{\omega_{j,\theta_j}\}_{j=1,\cdots,N})$,\\
(ii-1) there exists a subcentral measure $m$ on $E_{\mathfrak{A}\otimes_{\mathrm{m}}C_0(A)}$ such that
$\mu_{j,\theta_j}^{(A,\alpha)}\ll m$ $(j=1,2,\cdots,N)$, where $\mu_{j,\theta_j}^{(A,\alpha)}$ is the
central measure of $\omega_{j,\theta_j,(A,\alpha)}$ $(j=1,2,\cdots,N)$;\\
(ii-2) the support of $\mu_{j,\theta_j}^{(A,\alpha)}$ is independent of $\theta_j\in\Theta_j$.\\
The triplet $\mathfrak{P}=(\varphi,(A,\alpha),m)$ is called a predictive measurement for
$\{\omega_{1,\theta_1}|\theta_1\in\Theta_1\}$, $\cdots$, $\{\omega_{N,\theta_N}|\theta_N\in\Theta_N\}$.
\end{defi}
It is easily seen that
$\{\omega_{1,\theta_1}|\theta_1\in\Theta_1\}$, $\{\omega_{2,\theta_2}|\theta_2\in\Theta_2\}$
and $\{\omega_{3,\theta_3}|\theta_3\in\Theta_3\}$ are not necessarily comparable,
even if $\{\omega_{1,\theta_1}|\theta_1\in\Theta_1\}$ and $\{\omega_{2,\theta_2}|\theta_2\in\Theta_2\}$
are comparable and so are $\{\omega_{2,\theta_2}|\theta_2\in\Theta_2\}$
and $\{\omega_{3,\theta_3}|\theta_3\in\Theta_3\}$.

Let $\{\omega_{1,\theta_1}|\theta_1\in\Theta_1\},\cdots, \{\omega_{N,\theta_N}|\theta_N\in\Theta_N\}$ be
predictively comparable and $\mathfrak{P}=(\varphi,(A,\alpha),m)$ a predictive measurement for them.
For $j=1,\cdots, N$,
\begin{eqnarray*}
\omega_{j,\theta_j,(A,\alpha)}
&=&\alpha^\ast(\tilde{\omega}_{j,\theta_j}\otimes m_{\mathcal{A}})\circ (\pi_\omega\otimes id) \\
 &=:& \int\rho\;d\mu_{j,\theta_j}^{(A,\alpha)}(\rho)\\
 &=&\int\rho\;\frac{d\mu_{j,\theta_j}^{(A,\alpha)}}{dm}(\rho)\;dm(\rho),
\end{eqnarray*}
and put $\displaystyle{p^{(A,\alpha)}_{j,\theta_j}(\rho)=\frac{d\mu_{j,\theta_j}^{(A,\alpha)}}{dm}(\rho)}$.
We denote by $\rho^n=\{\rho_1,\rho_2,\cdots,\rho_n\}$ the set of states consisting of $n$ elements.
Then we define a maximal likelihood predictive state
$\omega_{\hat{\theta}_j(\rho^n)}$, where
$\displaystyle{\hat{\theta}_j(\rho^n)=\arg\!\max_{\theta_j\in\Theta_j}
\prod_{i=1}^n p^{(A,\alpha)}_{j,\theta_j}(\rho_i)}$ for $\rho^n=\{\rho_1,\rho_2,\cdots,\rho_n\}$,
and a Bayesian Escort predictive state
\begin{equation*}
\omega_{j,\pi_j,\beta}^{\rho^n}(B)=
\frac{\int \omega_{j,\theta_j}(B) \prod_{i=1}^n p^{(A,\alpha)}_{j,\theta_j}(\rho_i)^\beta\pi_j(\theta_j)d\theta_j}
{\int \prod_{i=1}^n p^{(A,\alpha)}_{j,\theta_j}(\rho_i)^\beta\pi_j(\theta_j)d\theta_j},
\end{equation*}
for $j=1,2,\cdots,N$, where $\pi_j(\theta_j)$ is a probability density on $\Theta_j$ and $\beta>0$.
It then holds that
\begin{eqnarray*}
\omega_{\hat{\theta}_j(\rho^n),(A,\alpha)} &=&
\int \rho\;p^{(A,\alpha)}_{j,\hat{\theta}_j(\rho^n)}(\rho)\;dm(\rho),\\
\omega_{j,\pi_j,\beta,(A,\alpha)}^{\rho^n} &=&\int\rho\;p^{\rho^n,(A,\alpha)}_{j,\pi_j,\beta}(\rho)\;dm(\rho),
\end{eqnarray*}
where
\begin{equation*}
p^{\rho^n,(A,\alpha)}_{j,\pi_j,\beta}(\rho)=
\frac{\int p^{(A,\alpha)}_{j,\theta_j}(\rho) \prod_{i=1}^n p^{(A,\alpha)}_{j,\theta_j}(\rho_i)^\beta\pi_j(\theta_j)d\theta_j}
{\int \prod_{i=1}^n p^{(A,\alpha)}_{j,\theta_j}(\rho_i)^\beta\pi_j(\theta_j)d\theta_j}.
\end{equation*}
Futhermore, we can show the following equalities by Theorem \ref{HOT}:
\begin{eqnarray*}
S(\varphi\Vert\omega_{\hat{\theta}_j(\rho^n)}) 
&=&\int q(\rho)\log \frac{q(\rho)}{p^{(A,\alpha)}_{j,\hat{\theta}_j(\rho^n)}(\rho)}\;dm(\rho)
=D(q\Vert p^{(A,\alpha)}_{j,\hat{\theta}_j(\rho^n)}),\\
S(\varphi\Vert\omega_{j,\pi_j,\beta}^{\rho^n}) 
&=&\int q(\rho)\log\frac{q(\rho)}{p^{\rho^n,(A,\alpha)}_{j,\pi_j,\beta}(\rho)}\;dm(\rho)
=D(q\Vert p^{\rho^n,(A,\alpha)}_{j,\pi_j,\beta}),
\end{eqnarray*}
where $q(\rho)=\dfrac{d\mu_{\varphi_{(A,\alpha)}}}{dm}(\rho)$.
Thus we can define information criteria for quantum states, owing to
discussion on the formulation of that for probability distributions \cite{Ak74,W1,W2,W3}.
Information criteria for quantum states, such as
\begin{eqnarray*}
\mathrm{AIC}_j&=&-\frac{1}{n}\sum_{i=1}^n\log p^{(A,\alpha)}_{j,\hat{\theta}_j(\rho^n)}(\rho_i)+\frac{d_j}{n},\\
 \mathrm{WAIC}_j&=&
 -\frac{1}{n}\sum_{i=1}^n\log p^{\rho^n,(A,\alpha)}_{j,\pi_j,\beta}(\rho_i)+\frac{\beta}{n}
 \mathcal{V}^{(A,\alpha)}_j,
\end{eqnarray*}
are defined for $j=1,\cdots,N$, where $d_j=\dim \Theta_j$,
and $0<\beta<\infty$, the functional variance $\mathcal{V}^{(A,\alpha)}_j$ defined by
\begin{equation}\label{FV}
\mathcal{V}^{(A,\alpha)}_j=\sum_{i=1}^{n}\left\{  \langle(\log p_{j,\theta_j}^{(A,\alpha)}(\rho_{i}))^{2}
\rangle_{\pi_j,\beta}^{\rho^{n}}-(\langle\log p_{j,\theta_j}^{(A,\alpha)}(\rho_{i})\rangle
_{\pi_j,\beta}^{\rho^{n}})^{2}\right\}.
\end{equation}
The a posteriori mean $\langle\cdot\rangle_{\pi_j,\beta}^{\rho^{n}}$ in Eq.(\ref{FV}) is defined by
\begin{equation}
\langle G(\theta_j)\rangle_{\pi_j,\beta}^{\rho^{n}}=\frac{\int
G(\theta_i)\prod_{i=1}^{n}p_{j,\theta_j}^{(A,\alpha)}(\rho_{i})^{\beta}\pi_j(\theta_j)d\theta_j
}{\int\prod_{i=1}^{n}p_{j,\theta_j}^{(A,\alpha)}(\rho_{i})^{\beta}\pi_j
(\theta_j)d\theta_j},
\end{equation}
for a given function $G(\theta_j)$ on $\Theta_j$.
We can choose the best model from
$\{\omega_{1,\theta_1}|\theta_1\in\Theta_1\},\cdots, \{\omega_{N,\theta_N}|\theta_N\in\Theta_N\}$
by information criteria in the setting of a predictive measurement for them.
Furthermore, it is obvious by the above discussion that the accuracy of the estimation
for quantum states are the same as that for probability measures in the classical case.
Thus we have demonstrated the validity of the use of information criteria for models of quantum states.
On the other hand, we can conclude that what makes applications of information criteria
for quantum states difficult is whether measurements which can compare all models of states
under consideration exist or not.

\section{On the Quantum $\alpha$-Divergence}
We discuss the quantum $\alpha$-divergence and its application here.
\subsection{The $\alpha$-version of Hiai-Ohya-Tsukada Theorem}
We define the $\alpha$-divergence
\begin{equation*}
D^{(\alpha)}(\mu\Vert\nu)=
\left\{
\begin{array}{ll}
D(\nu\Vert\mu), \;\;  (\alpha=+1), &\\
\!\!\!\dfrac{4}{1-\alpha^2}\!\left(\displaystyle{1-\!\!\int dm \left(\frac{d\mu}{dm}\right)^{\frac{1-\alpha}{2}}
\!\!\left(\frac{d\nu}{dm}\right)^{\frac{1+\alpha}{2}}}\right)\!, & \\
\hspace{14mm}\;\; (|\alpha|<1), & \\
D(\mu\Vert\nu), \;\; (\alpha=-1), &
\end{array}
\right. 
\end{equation*}
for probability measures $\mu,\nu$ on a measurable space $(\Omega,\mathcal{F})$
which are absolutely continuous with respect to a measure $m$ on $(\Omega,\mathcal{F})$,
Uhlmann's $\alpha$-divergence
\begin{equation*}
S^{(\alpha)}(\varphi\Vert\psi)_{\mathrm{Uhlmann}}=
\left\{
\begin{array}{ll}
S(\psi\Vert\varphi)_{\mathrm{Uhlmann}},\;\;(\alpha=+1), & \\
\!\!\!\dfrac{4}{1-\alpha^2}(1-QF_{\frac{1+\alpha}{2}}(\varphi^R,\psi^L)(1,1)), & \\
\hspace{27mm}(|\alpha|<1), & \\
S(\varphi\Vert\psi)_{\mathrm{Uhlmann}},\;\;(\alpha=-1), &\
\end{array}
\right.
\end{equation*}
for states $\varphi,\psi$ on a $\ast$-algebra, and Araki's $\alpha$-divergence
\begin{equation*}
S^{(\alpha)}(\varphi\Vert\psi)_{\mathrm{Araki}}=
\left\{
\begin{array}{ll}
S(\psi\Vert\varphi)_{\mathrm{Araki}},\;\;(\alpha=+1), & \\
\dfrac{4}{1-\alpha^2}
(1-\langle \Psi,\Delta_{\Phi,\Psi}^{1-\frac{1+\alpha}{2}}\Psi\rangle), & \\
\hspace{22mm}(|\alpha|<1), & \\
S(\varphi\Vert\psi)_{\mathrm{Araki}},\;\;(\alpha=-1), &\
\end{array}
\right.
\end{equation*}
for normal states $\varphi=\langle \Phi, \cdot\;\Phi\rangle,\psi=\langle \Psi, \cdot\;\Psi\rangle$
on a $\sigma$-finite von Neumann algebra $\mathcal{M}$, where $\Phi,\Psi$ are elements of a self-dual cone of a
standard form. It is easily checked that
\begin{equation*}
S^{(\alpha)}(\mu\Vert\nu)_{\mathrm{Araki}}=D^{(\alpha)}(\mu\Vert\nu),
\end{equation*}
for $\mu.\nu\ll m$ on a von Neumann algebra $L^\infty(\Omega,m)$.

The following three theorems are the $\alpha$-analogue of theorems in subsection 3.1.
Proofs are given in Appendix A.
\begin{theo}
For any positive linear functionals $\varphi,\psi$ on a von Neumann algebra $\mathcal{M}$,
$S^{(\alpha)}(\varphi\Vert\psi)_{\mathrm{Uhlmann}}=S^{(\alpha)}(\varphi\Vert\psi)_{\mathrm{Araki}}$.
\end{theo}
By this theorem, it is not necessary to distinguish
Araki's $\alpha$-divergence and Uhlmann's one.
We call them the quantum $\alpha$-divergence, and denote them by $S^{(\alpha)}(\varphi\Vert\psi)$.
\begin{theo}
Let $\varphi,\psi$ be positive linear functionals on a C$^\ast$-algebra $\mathfrak{A}$
and $\pi$ be a nondegenerate representation of $\mathfrak{A}$ on a Hilbert space.
If there are the normal extentions $\tilde{\varphi},\tilde{\psi}$ of $\varphi,\psi$ to
$\pi(\mathfrak{A})^{\prime\prime}$ such that
$\varphi(A)=\tilde{\varphi}(\pi(A)),\psi(A)=\tilde{\psi}(\pi(A))$, then
$S^{(\alpha)}(\varphi\Vert\psi)=S^{(\alpha)}(\tilde{\varphi}\Vert\tilde{\psi})$.
\end{theo}

\begin{theo}
Let $\varphi,\psi$ be states on $\mathfrak{A}$ with barycentric measures $\mu,\nu$, respectively.
If there exists a subcentral measure $m$ such that $\mu,\nu\ll m$,
then $S^{(\alpha)}(\varphi\Vert\psi)=D^{(\alpha)}(\mu\Vert\nu)$.
\end{theo}

\subsection{A Quantum Version of Chernoff Bound}
We use the same notation in subsection 3.2.
\begin{eqnarray*}
\mathrm{M}(\omega_1,\omega_2) \!\!\!&=&\!\!\!\bigcup_{A=A^\ast\in\pi_{\omega}(\mathfrak{A})^{\prime\prime}}
\mathrm{M}(A; \omega_1,\omega_2), \\
\mathrm{M}(A; \omega_1,\omega_2) \!\!\!&=&\!\!\!
\{(A,\alpha) \in \pi_{\omega}(\mathfrak{A})^{\prime\prime}\times Aut(\pi_\omega(\mathfrak{A})^{\prime\prime}\otimes
 \mathcal{A})|\;\mathcal{A}=C_0(A)^{\prime\prime},\\
&&\mu_{\omega_{1,(A,\alpha)}}, \mu_{\omega_{2,(A,\alpha)}}\ll \exists m:\;\text{a}\;\text{subcentral}\;
\text{measure}\},
\end{eqnarray*}
For a $(A,\alpha)\in \mathrm{M}(\omega_1,\omega_2)$, we define
$\mathcal{R}^{(A,\alpha)}_n(T_n)=\pi_1\alpha^{(A,\alpha)}_n(T_n)+\pi_2\beta_n^{(A,\alpha)}(T_n)$,
where $\pi_1,\pi_2\in\mathbb{R}$ such that $\pi_1,\pi_2>0,\pi_1+\pi_2=1$.

\begin{theo}
\begin{equation*}
\lim_{n\rightarrow \infty}\inf_{T_n}
\frac{1}{n}\log \mathcal{R}^{(A,\alpha)}_n(T_n)=\inf_{0\leq t\leq 1}\log F_{t}(\varphi,\psi),
\end{equation*}
where $F_{t}(\varphi,\psi)=QF_{t}(\varphi^R,\psi^L)(1,1)$.
\end{theo}
\begin{proof}
By \cite[Theorem 2.1]{NS09},
\begin{equation*}
\lim_{n\rightarrow \infty}\inf_{T_n}
\frac{1}{n}\log \mathcal{R}^{(A,\alpha)}_n(T_n)=
\inf_{0\leq t\leq 1}\log
\int dm(p_{\varphi}^{(A,\alpha)} )^{1-t}
(p_{\psi}^{(A,\alpha)} )^{t},
\end{equation*}
where $p_{\varphi}^{(A,\alpha)}=\dfrac{d\mu_{\varphi,(A,\alpha)}}{dm}$,
$p_{\psi}^{(A,\alpha)}=\dfrac{d\mu_{\psi,(A,\alpha)}}{dm}$.
\begin{eqnarray*}
\int dm\;(p_{\varphi}^{(A,\alpha)} )^{1-t}(p_{\psi}^{(A,\alpha)} )^{t}
&=&QF_{t}(\varphi_{(A,\alpha)}^R,\psi_{(A,\alpha)}^L)(1\otimes_{\mathrm{m}} 1,1\otimes_{\mathrm{m}} 1) \\
 &=&QF_{t}(\varphi^R,\psi^L)(1,1)=F_{t}(\varphi,\psi).
\end{eqnarray*}

\end{proof}

\subsection{Bayesian $\alpha$-Predictive State}
The concept of a Bayesian $\alpha$-predictive state or a generalied Bayes predictive state
studied in earlier papers \cite{Ai75,CG99,T06,TK05}
is a generalization of a Bayes (escort) predictive state which appeared in the previous section.
However, we can easily see that it is difficult to apply this concept in a C$^\ast$-algebraic setting.
Thus we define a class of statistical model.
\begin{defi}
A family of states $\{\omega_{\theta}|\theta\in\Theta\}$ parametrized by a
compact set $\Theta$ in $\mathbb{R}^{d}$ is called a classical model if it
satisfies the following three conditions: \newline$\;\;(i)$ There is a
subcentral measure $m$ on $E_{\mathfrak{A}}$ such that $\mu_{\omega_{\theta}%
}\ll m$ for every $\theta\in\Theta$.\newline$\;(ii)$ The set
$\displaystyle{\overline{\left\{  \rho\in E_{\mathfrak{A}}\Big|p_{\theta}(\rho):=\frac
{d\mu_{\omega_{\theta}}}{dm}(\rho)>0\right\}  }}$ is independent of $\theta
\in\Theta$.\newline$(iii)$ $\omega_{\theta}$ is Bochner integrable.
\end{defi}
Let us define a Bayesian $\alpha$-predictive state.
\begin{defi}
Let $\{\omega_{\theta}|\theta\in\Theta\}$ be a classical model,
$\pi(\theta)$ be a probability density on $\Theta$, $\rho^n=\{\rho_1,\cdots,\rho_n\}$,
$-1\leq \alpha\leq 1$.
A state
\begin{equation}
\omega_{\pi,\alpha}^{\rho^n}=\frac{1}{C_{\pi,\alpha}^{\rho^n}} \int\rho\;p_{\pi,\alpha}(\rho|\rho^n)\;dm(\rho),
\end{equation}
is called a Bayesian $\alpha$-predictive state,
where 
\begin{equation*}
p_{\pi,\alpha}(\rho|\rho^n)=
\left\{
\begin{array}{ll}
\left(\int p_\theta(\rho)^{\frac{1-\alpha}{2}}
\pi(\theta|\rho^n)d\theta\right)^{\frac{2}{1-\alpha}} , &\; (\alpha\neq 1) \\
\exp\left(\int \log p_\theta(\rho)\pi(\theta|\rho^n)d\theta  \right) , &\;(\alpha=1)
\end{array}
\right.
\end{equation*}
a posteriori probability density
\begin{equation}
\pi(\theta|\rho^n)=\frac{\pi(\theta)\prod_{j=1}^n p_\theta(\rho_j)}{\int\pi(\theta)\prod_{j=1}^n p_\theta(\rho_j)
d\theta},
\end{equation}
and
\begin{equation*}
C_{\pi,\alpha}^{\rho^n}=\int p_{\pi,\alpha}(\rho|\rho^n)\;dm(\rho).
\end{equation*}
\end{defi}
We prove the following theorem to justify the use of a Bayesian $\alpha$-predictive state
in the context of quantum statistical decision theory, which is first given by \cite{Ai75}
and generalized by \cite{CG99,TK05,T06}.
\begin{theo}
For a state-valued function $\rho^n\mapsto \varphi^{\rho^n}$ such that it has a barycentric measure
which is absolutely continuous with respect to $m$, its risk function
\begin{equation}
R^{(\alpha)}(\omega_\theta\Vert\varphi^{\cdot})=\int\!\!\!\int S^{(\alpha)}(\omega_\theta\Vert\varphi^{\rho^n})
\prod_{i=1}^n d\mu_{\theta}(\rho_i) \pi(\theta)d\theta
\end{equation}
is minimized at $\omega_{\pi,\alpha}^{\rho^n}$.
\end{theo}
\begin{proof}
We prove this theorem only when $\alpha\neq\pm 1$, since $\alpha=\pm 1$ can be proved by similar argument.
\begin{equation*}
S^{(\alpha)}(\omega_\theta\Vert\varphi^{\rho^n})=
\frac{4}{1-\alpha^2}\int \left( 1-
p_\theta(\rho)^{\frac{1-\alpha}{2}}q^{\rho^n}(\rho)^{\frac{1+\alpha}{2}}\right) dm(\rho),
\end{equation*}
where $q^{\rho^n}=\dfrac{d\mu_{\varphi^{\rho^n}}}{dm}$.
\begin{eqnarray}
&& R^{(\alpha)}(\varphi^{\cdot}_1\Vert\omega_\theta)-R^{(\alpha)}(\varphi^{\cdot}_2\Vert\omega_\theta) \nonumber\\
 &=& \int\!\!\!\int \left\{\frac{4}{1-\alpha^2}\int \left( p_\theta(\rho)^{\frac{1-\alpha}{2}}
 q^{\rho^n}_2(\rho)^{\frac{1+\alpha}{2}}-p_\theta(\rho)^{\frac{1-\alpha}{2}}
 q^{\rho^n}_1(\rho)^{\frac{1+\alpha}{2}}\right) dm(\rho)
 \right\}\prod_{i=1}^n d\mu_{\theta}(\rho_i) \pi(\theta)d\theta\nonumber \\
 &=&\frac{4}{1-\alpha^2}\int\!\!\!\int\!\!\!\int p_\theta(\rho)^{\frac{1-\alpha}{2}}
 (q^{\rho^n}_2(\rho)^{\frac{1+\alpha}{2}}-q^{\rho^n}_1(\rho)^{\frac{1+\alpha}{2}})
 dm(\rho)\prod_{i=1}^n p_\theta(\rho_i)dm(\rho_i) \pi(\theta)d\theta\nonumber \\
 &=&\frac{4}{1-\alpha^2}\int\!\!\!\int
 \left(\int p_\theta(\rho)^{\frac{1-\alpha}{2}} \prod_{i=1}^n p_\theta(\rho_i)\pi(\theta)d\theta\right)
 (q^{\rho^n}_2(\rho)^{\frac{1+\alpha}{2}}-q^{\rho^n}_1(\rho)^{\frac{1+\alpha}{2}})dm(\rho)
 \prod_{i=1}^n dm(\rho_i)\nonumber \\
 &=&\frac{4}{1-\alpha^2}\int (C_{\pi,\alpha}^{\rho^n})^{\frac{1-\alpha}{2}} \left\{\int
 p_{\pi,\alpha}(\rho|\rho^n)^{\frac{1-\alpha}{2}}
(q^{\rho^n}_2(\rho)^{\frac{1+\alpha}{2}}-q^{\rho^n}_1(\rho)^{\frac{1+\alpha}{2}})dm(\rho)\right\}
 p_\pi(\rho^n) \prod_{i=1}^n dm(\rho_i) \nonumber
\end{eqnarray}
where $\displaystyle{p_\pi(\rho^n)=\int \prod_{i=1}^n p_\theta(\rho_i)\pi(\theta)d\theta}$.
When $\varphi^{\rho^n}_2$ is equal to $\omega_{\pi,\alpha}^{\rho^n}$,
\begin{eqnarray}
 && R^{(\alpha)}(\varphi^{\cdot}_1\Vert\omega_\theta)
 -R^{(\alpha)}(\omega_{\pi,\alpha}^{\cdot}\Vert\omega_\theta) \nonumber\\
 &=&\!\!\!\frac{4}{1-\alpha^2}\int (C_{\pi,\alpha}^{\rho^n})^{\frac{1-\alpha}{2}}
 \left\{\int p_{\pi,\alpha}(\rho|\rho^n)^{\frac{1-\alpha}{2}}
(p_{\pi,\alpha}(\rho|\rho^n)^{\frac{1+\alpha}{2}}-q^{\rho^n}_1(\rho)^{\frac{1+\alpha}{2}})dm(\rho)\right\}
 p_\pi(\rho^n) \prod_{i=1}^n dm(\rho_i) \nonumber \\
 &=&\int (C_{\pi,\alpha}^{\rho^n})^{\frac{1-\alpha}{2}} \left\{\frac{4}{1-\alpha^2}\int
 (1-p_{\pi,\alpha}(\rho|\rho^n)^{\frac{1-\alpha}{2}}q^{\rho^n}_1(\rho)^{\frac{1+\alpha}{2}})dm(\rho)\right\}
 p_\pi(\rho^n) \prod_{i=1}^n dm(\rho_i) \nonumber \\
 &=&\int (C_{\pi,\alpha}^{\rho^n})^{\frac{1-\alpha}{2}} D^{(\alpha)}
 (p_{\pi,\alpha}(\rho|\rho^n)\Vert q^{\rho^n}_1(\rho))
 p_\pi(\rho^n) \prod_{i=1}^n dm(\rho_i) \nonumber \\
 &=&\int (C_{\pi,\alpha}^{\rho^n})^{\frac{1-\alpha}{2}} S^{(\alpha)}
 (\omega_{\pi,\alpha}^{\rho^n}\Vert \varphi^{\rho^n}_1)
 p_\pi(\rho^n) \prod_{i=1}^n dm(\rho_i)\geq 0. \nonumber
\end{eqnarray}
The last equality holds when $\varphi^{\rho^n}_1$ equals to $\omega_{\pi,\alpha}^{\rho^n}$,
since $\varphi^{\rho^n}_1$ is arbitrary. This means
\begin{equation*}
R^{(\alpha)}(\varphi^{\cdot}\Vert\omega_\theta)
\geq R^{(\alpha)}(\omega_{\pi,\alpha}^{\cdot}\Vert\omega_\theta)
\end{equation*}
for any state-valued function $\varphi^{\rho^n}$.
\end{proof}

\section{Conclusion}
We have proved the validity of the use of Hiai-Ohya-Tsukada theorem on the basis of sector theory
and measuring processes. As a result, a quantum version of Stein's lemma and of Sanov's theorem
are proved, and it turns out that the quantum relative entropy is a rate fuction in large deviarion theory.
In addition, we have formulated the measurement which allows us to apply information criteria, and defined
information criteria for quantum states. It is shown that their accuracy is the same as that in classical case.
However, there is plenty of room for deepening measurements which we can evalulate for each model equivalently.

Let us compare the results in this paper with past studies.
The methods in this paper are extensions of classical ones,
which suit modern statistics and stand different views from \cite{He76,Ho82} and succeeding investigations.
In particular, model selection using information criteria compares with hypothesis testing as methods for
constructing models and testing hypotheses, and is expected to apply to quantum systems effectively.
In asymptotic theory of quantum hypothesis testing, we could consider the universal situation
increasing the number of measured data in this paper, while the special but important situation that
the number of quanta in the system increases were examined in past studies
\cite{Au07,Ha07,HN07,HP91,N06,NS09,OH04,ON00}.
It is common for both of them that optimization of measurement is vital. We hope that both universal
and special methods in quantum statistical inference will be developed in the future.

\appendix
\section{Proofs of Theorems}
\begin{proof}[Proof of Theorem 7]
It suffices to prove for $|\alpha|<1$.
It is proved by discussion in \cite[p.129]{HOT83} that
\begin{equation*}
QF_t(\varphi^R,\psi^L)(1,1)=\langle \Psi,\Delta_{\Phi,\Psi}^{1-t} \Psi\rangle.
\end{equation*}
Therefore, we have
\begin{eqnarray*}
S^{(\alpha)}(\varphi\Vert\psi)_{\mathrm{Uhlmann}}
&=&\frac{4}{1-\alpha^2}(1-QF_{\frac{1+\alpha}{2}}(\varphi^R,\psi^L)(1,1)) \\
 &=&\frac{4}{1-\alpha^2}(1-\langle \Psi,\Delta_{\Phi,\Psi}^{1-\frac{1+\alpha}{2}} \Psi\rangle) \\
 &=&S^{(\alpha)}(\varphi\Vert\psi)_{\mathrm{Araki}}.
\end{eqnarray*}
\end{proof}

\begin{proof}[Proof of Theorem 8]
It suffices to prove for $|\alpha|<1$.
By \cite[Lemma 3.1]{HOT83}, it is proved that
\begin{equation*}
QF_t(\varphi^R,\psi^L)(A,A)
=QF_t(\tilde{\varphi}^R,\tilde{\psi}^L)(\pi(A),\pi(A)),
\end{equation*}
where $\pi=\pi_{\varphi+\psi}$. Thus we have
\begin{eqnarray*}
S^{(\alpha)}(\varphi\Vert\psi) &=&\frac{4}{1-\alpha^2}(1-QF_{\frac{1+\alpha}{2}}(\varphi^R,\psi^L)(1,1)) \\
 &=&\frac{4}{1-\alpha^2}(1-QF_{\frac{1+\alpha}{2}}(\tilde{\varphi}^R,\tilde{\psi}^L)(\pi(1),\pi(1))) \\
 &=& S^{(\alpha)}(\tilde{\varphi}\Vert\tilde{\psi}).
\end{eqnarray*}
\end{proof}

\begin{lemm2}\label{mono}
Let $\varphi|_{\mathcal{B}}, \psi|_{\mathcal{B}}$ be the restrictions of states $\varphi,\psi$
on a von Neumann algebra $\mathcal{M}$ to a $\ast$-subalgebra $\mathcal{B}$ to $\mathcal{M}$. Then,
\begin{equation*}
S^{(\alpha)}_{\mathcal{B}}(\varphi\Vert\psi):=S^{(\alpha)}(\varphi|_{\mathcal{B}}\Vert\psi|_{\mathcal{B}})
\leq S^{(\alpha)}(\varphi\Vert\psi).
\end{equation*}
\end{lemm2}
\begin{proof}
It suffices to prove it for $|\alpha|<1$.
$|_{\mathcal{B}}:E_\mathcal{M}\rightarrow E_\mathcal{B}$ is the dual of the embedding
$\mathcal{B}\ni B \mapsto B\in\mathcal{M}$.
Using \cite[Proposition 9]{U77} and Theorem 8,
\begin{eqnarray*}
QF_t(\varphi^R\Vert\psi^L)(1,1)&=&QI_t(\varphi^R\Vert\psi^L)(1)^2 \\
 &\leq&QI_t((\varphi|_{\mathcal{B}})^R\Vert(\psi|_{\mathcal{B}})^L)(1)^2 \\
 &=&QF_t((\varphi|_{\mathcal{B}})^R\Vert(\psi|_{\mathcal{B}})^L)(1,1).
\end{eqnarray*}
Thus we have
\begin{eqnarray*}
S^{(\alpha)}(\varphi\Vert\psi)
&=&\frac{4}{1-\alpha^2}(1-QF_{\frac{1+\alpha}{2}}(\varphi^R,\psi^L)(1,1)) \\
 &\geq&
 \frac{4}{1-\alpha^2}(1-QF_{\frac{1+\alpha}{2}}((\varphi|_{\mathcal{B}})^R\Vert(\psi|_{\mathcal{B}})^L)(1,1)) \\
 &=& S^{(\alpha)}(\varphi|_{\mathcal{B}}\Vert\psi|_{\mathcal{B}})=S^{(\alpha)}_{\mathcal{B}}(\varphi\Vert\psi).
\end{eqnarray*}

\end{proof}
\begin{proof}[Proof of Theorem 9]
It suffices to prove for $|\alpha|<1$.

We define $\Gamma: \mathfrak{A}\rightarrow C(E_\mathfrak{A})$ by $\Gamma(A)(\omega)=\omega(A)$, and its dual map
$\Gamma^\ast:M_1(E_\mathfrak{A})\rightarrow E_\mathfrak{A}$ by
\begin{equation*}
(\Gamma^\ast\lambda)(A)=\int \omega(A)\;d\lambda(\omega)=\lambda(\Gamma(A)),
\end{equation*}
for $\lambda\in M_1(E_\mathfrak{A})$. $\Gamma$ satisfies
\begin{equation*}
\Gamma(1)=1,\;\Gamma(A^\ast) =\Gamma(A)^\ast,\;
\Gamma(A^\ast)\Gamma(A) \leq \Gamma(A^\ast A),
\end{equation*}
for any $A\in\mathfrak{A}$. Thus,
\begin{equation*}
S^{(\alpha)}(\varphi\Vert\psi) =S^{(\alpha)}(\Gamma^\ast\mu\Vert\Gamma^\ast\nu) 
\leq D^{(\alpha)}(\mu\Vert\nu).
\end{equation*}
By assumption, $m$ is a subcentral measure $\mu_{\mathcal{C}}$ associated with a von Neumann
subalgebra $\mathcal{C}$ of the center
$\mathfrak{Z}_\chi(\mathfrak{A})=\pi_{\chi}(\mathfrak{A})^{\prime\prime}\cap\pi_{\chi}(\mathfrak{A})^{\prime}$
of a state $\chi\in E_{\mathfrak{A}}$.
We define a $\ast$-isomorphism $\kappa_m:L^\infty(E_\mathfrak{A},m)\rightarrow \mathcal{C}$ by
\begin{equation*}
\langle \Omega_\chi,\pi_\chi(A)\kappa_m(f) \Omega_\chi\rangle=\int f(\rho) \rho(A)\;dm(\rho).
\end{equation*}
It is then proved in \cite[Example 2.6 and Theorem 3.2]{HOT83} that
\begin{eqnarray*}
\tilde{\varphi}(\kappa_m(f))&=&\int f(\rho)\;d\mu(\rho),\\
\tilde{\psi}(\kappa_m(f))&=&\int f(\rho)\;d\nu(\rho),
\end{eqnarray*}
where $\tilde{\varphi}, \tilde{\psi}$ are the normal extensions of $\varphi, \psi$ to
$\pi_{\chi}(\mathfrak{A})^{\prime\prime}$. By Lemma \ref{mono},
\begin{eqnarray*}
S^{(\alpha)}(\varphi\Vert\psi) &=&S^{(\alpha)}(\tilde{\varphi}\Vert\tilde{\psi}) \\
 &\geq&S^{(\alpha)}_{\mathcal{C}}(\tilde{\varphi}\Vert\tilde{\psi})=D^{(\alpha)}(\mu\Vert\nu).
\end{eqnarray*}
\end{proof}

\section*{Acknowledgment}
The author would like to thank Prof. Izumi Ojima, Dr. Hayato Saigo, Dr. Takahiro Hasebe,
Dr. Hiroshi Ando for their warm encouragement and useful comments.

\end{document}